\documentclass[journal]{IEEEtran}

\usepackage[OT1]{fontenc}
\usepackage[numbers,sort&compress]{natbib}
\usepackage[cmex10]{amsmath}
\usepackage{amssymb}
\usepackage{bm}
\usepackage{graphicx}
\usepackage{subfigure}
\usepackage{tabularx}
\usepackage{arydshln}
\usepackage{mathtools}
\usepackage{multirow}
\usepackage{algorithm,algorithmic}
\usepackage{url}
\usepackage{listings}
\usepackage{comment}
\usepackage{latexsym}
\usepackage{amsthm}
\usepackage{braket}
\usepackage{physics}
\usepackage{mleftright}
\mleftright

\def\c{\mathbf{c}}
\def\C{\mathcal{C}}
\def\etal{\textit{et al.}~}

\newtheorem{thm}{Theorem}
\newtheorem{lem}{Lemma}
\newtheorem{cor}{Corollary}

\newcommand{\argmin}{\mathop{\mathrm{arg~min}}\limits}

\begin{document}

\title{Quantum Search Algorithm for Binary Constant Weight Codes}

\author{Kein~Yukiyoshi and Naoki~Ishikawa,~\IEEEmembership{Senior~Member,~IEEE}.\thanks{K.~Yukiyoshi and N.~Ishikawa are with the Faculty of Engineering, Yokohama National University, 240-8501 Kanagawa, Japan (e-mail: ishikawa-naoki-fr@ynu.ac.jp). This research was partially supported by the Japan Society for the Promotion of Science KAKENHI (Grant Number JP22H01484).}}

\markboth{\today}
{Shell \MakeLowercase{\textit{et al.}}: Bare Demo of IEEEtran.cls for Journals}
\maketitle

\begin{abstract}
A binary constant weight code is a type of error-correcting code with a wide range of applications. The problem of finding a binary constant weight code has long been studied as a combinatorial optimization problem in coding theory. In this paper, we propose a quantum search algorithm for binary constant weight codes. Specifically, the search problem is newly formulated as a quadratic unconstrained binary optimization (QUBO) and Grover adaptive search (GAS) is used for providing the quadratic speedup. Focusing on the inherent structure of the problem, we derive an upper bound on the minimum of the objective function value and a lower bound on the exact number of solutions. In our algebraic analysis, it was found that this proposed algorithm is capable of reducing the number of required qubits, thus enhancing the feasibility. Additionally, our simulations demonstrated that it reduces the query complexities by $63\%$ in the classical domain and $31\%$ in the quantum domain. The proposed approach may be useful for other quantum search algorithms and optimization problems.
\end{abstract}

\IEEEpeerreviewmaketitle


\section{Introduction}
\IEEEPARstart{A} binary constant weight code is an error-correcting code, defined such that all codewords have the same Hamming weight \cite{brouwer1990new}. The codes have a wide range of applications, including the design of very large-scale integration systems \cite{tallini1998design}, the frequency assignment problem \cite{moon2005assignment}, the design of optical orthogonal codes \cite{chung1989optical}, error-tolerant key generation from noisy data \cite{dodis2008fuzzy}, inter-cell interference mitigation for flash memory \cite{yassine2017index}, and the design of Grassmannian constellations mainly used in wireless communications \cite{delacruz2021new}.

A binary constant weight code $\C$ is composed of $M$ different codewords, each of which is a binary row vector $\c = [c_0 ~ \cdots ~ c_{n-1}]$ with a length $n$ and weight $w$.
Let $d(\c_m, \c_{m'})$ be the Hamming distance between binary vectors $\c_m$ and $\c_{m'}$.
Its minimum value $d = \min_{0 \leq m < m' < M} d(\c_m, \c_{m'})$ is termed the minimum distance of $\C$.
The maximum number of possible codewords is denoted by $A(n, d, w)$, which is an upper bound on $M$.
The main concern in the coding theory \cite{johnson1962new,brouwer1980few,nurmela1997new,smith2006new,montemanni2009heuristic,etzion2014new} is to clarify $A(n, d, w)$ or to construct a code that achieves $A(n, d, w)$.
However, depending on the parameters $(n, d, w)$, the optimization of binary constant weight codes becomes a challenging task due to its combinatorial explosion.

Quantum computation has proven to have advantages over classical computation regarding query complexity required for solving specific problems, such as Shor's algorithm \cite{shor1994algorithms} and Grover's algorithm \cite{grover1996fast}.
Grover's search algorithm finds a solution in an unordered database of $N$ elements with query complexity $O(\sqrt{N})$.
Boyer \etal extended Grover's algorithm to the case where the number of desired solutions is initially unknown \cite{BOYER1998TIGHT}, which is termed the Boyer-Brassard-H\o yer-Tapp (BBHT) algorithm.
Then, D\" urr and H\o yer proposed Grover adaptive search (GAS) \cite{DURR1999QUANTUM,BULGER2003IMPLEMENTING,BARITOMPA2005GROVER}, an extension of the BBHT algorithm to solve optimization problems.
GAS can be interpreted as a kind of quantum exhaustive search algorithm, which guarantees the optimality of the solutions obtained.
At the time of writing, it is the only algorithm that provides the quadratic speedup for solving a binary optimization problem.
An efficient construction method for a quantum circuit corresponding to a polynomial objective function has long been unknown.
The breakthrough is Gilliam's method proposed in \cite{GILLIAM2021GROVER} that efficiently constructs a quantum circuit corresponding to an objective function of a quadratic unconstrained binary optimization (QUBO) problem. This method has already demonstrated quadratic speedups of maximum likelihood detection \cite{norimoto2022quantum} and wireless channel assignment problem \cite{sano2022qubit}.
Note that it also supports higher-order terms in addition to quadratic terms and is a game-changing innovation. 

Against this background, we newly formulate the problem of finding a binary constant weight code as a QUBO problem.
This novel formulation enables the use of GAS and provides a quadratic speedup for the search problem.
Additionally, focusing on the inherent structure of the problem, our proposed GAS further accelerates the quadratic speedup, which is a unique attempt in the literature.
The major contributions of this paper are summarized as follows:
\begin{enumerate}
    \item The problem of finding an optimal binary constant weight code is newly formulated as a QUBO problem, and the number of qubits required for constructing a quantum circuit is maximally reduced. In this formulation, the Hamming distance between codewords, which is usually represented by exclusive OR (XOR), is attributed to the representation using an inner product. Furthermore, a unique trick using the inner product is incorporated into the objective function to guarantee that the minimum distance of an obtained code is at least a given value $d$.
    \item Two novel mathematical properties are clarified: (1) an upper bound on the minimum of the objective function value and (2) a lower bound on the exact number of solutions. These clarifications help reduce the number of qubits and query complexity, which are supported by our algebraic analysis and numerical simulations.
\end{enumerate}

The remainder of this paper is organized as follows:
In Section~II, we review conventional quantum search algorithms such as BBHT and GAS.
In Section~III, we present our novel QUBO formulation, and our novel algorithm based on GAS.
In Section~IV, the proposed GAS is compared with the conventional GAS, where both rely on the proposed formulation.
Finally, in Section~V, we conclude this paper.

Italicized symbols represent scalar values, and bold symbols represent vectors and matrices.
We use zero-based indexing.
Table~\ref{table:sym} summarizes the important mathematical symbols used in this paper.

\begin{table}[tb]
        \centering
        \caption{List of important mathematical symbols. \label{table:sym}}
        \begin{tabular}{lll}
            $\mathbb{F}_2 = \{0, 1\}$ & & Finite field of order 2 \\
            $\mathbb{R}$ & & Real numbers \\
            $\mathbb{C}$ & & Complex numbers\\
            $\mathbb{Z}$ & & Integers \\
            $j$ & $\in \mathbb{C}$ & Imaginary number\\
            $(\cdot)^\mathrm{T}$ & & Matrix transpose\\
            $\overline{(\cdot)}$ & $\in \mathbb{R}$ & Upper bound\\
            $\underline{(\cdot)}$ & $\in \mathbb{R}$ & Lower bound\\
            $\C$ & & Code \\
            $\c$ & $\in \mathbb{F}_2^{1 \times n}$ & Codeword\\
            $d(\c_m, \c_{m'})$ & $\in \mathbb {Z}$ & Hamming distance between $\c_m$ and $\c_{m'}$\\
                $d$ & $\in \mathbb {Z}$ & Minimum distance of a code\\
                $n$ & $\in \mathbb {Z}$ & Length of a codeword\\
                $w$ & $\in \mathbb {Z}$ & Hamming weight\\
                $M$ & $\in \mathbb {Z}$ & Number of codewords\\
                $x$ & $\in \mathbb{F}_2$ & Binary variable \\
                $E(\cdot)$ & $\in \mathbb{Z}$ & Objective function\\
                $\rho$ & $\in \mathbb{Z}$ & Penalty coefficient of $E(\cdot)$\\
                $l$ & $\in \mathbb{Z}$ & Exponential coefficient of $E(\cdot)$\\
                $q_{1}$ & $\in \mathbb{Z}$ & Number of QUBO variables \\
                $q_{2}$ & $\in \mathbb{Z}$ & Number of qubits required for $E(\cdot)$\\
                $y_i$ & $\in \mathbb{Z}$ & Threshold for $E(\cdot)$ at $i$-th iteration \\
                $L_i$ & $\in \mathbb{Z}$ & Number of Grover operators \\
                $\mathbf{P}(n, w)$ & $\in \mathbb{F}_2^{\binom{n}{w} \times n}$ & Combinatorial matrix~\cite{frenger1999parallel}\\
                $\mathbf{p}_r$ & $\in \mathbb{F}_2^{1 \times n}$ & $(r+1)$-th row vector of $\mathbf{P}(n, w)$
        \end{tabular}
\end{table}

\section{Conventional Quantum Search Algorithms}
Quantum search algorithms are among the most important applications in quantum computation. In this section, we introduce the conventional quantum search and optimization algorithms upon which our proposed algorithm is based.

\subsection{Grover's Search and BBHT Algorithm \cite{BOYER1998TIGHT,BRASSARD2002QUANTUM}}
Grover's search algorithm~\cite{grover1996fast} finds one of $t$ solutions in $N$ elements.
The query complexity in the quantum domain is $O\qty(\sqrt{N / t})$, which is the total number of Grover operators applied to reach a desired solution.
Specifically, the probability of observing the desired solution is amplified by applying the Grover operator $\mathbf{G}$ to the uniform superposition state $\frac{1}{\sqrt{N}}\sum_i \Ket{i}$.
When the Grover operator is applied $L$ times, the success probability is expressed as \cite{BRASSARD2002QUANTUM}
\begin{equation}
    \label{eq:p_success}
    P_{\mathrm{success}}(L) = \sin^2\left((2L + 1)\arcsin \qty(\sqrt{\frac{t}{N}})\right).
\end{equation}
According to \eqref{eq:p_success}, the number of Grover operators that maximizes $P_{\mathrm{success}}(L)$ can be expressed as \cite{BRASSARD2002QUANTUM}
\begin{equation}
    L_{\mathrm{opt}} =  \left\lfloor \frac{\pi}{4}\sqrt{\frac{N}{t}} \right\rfloor.
    \label{eq:Lopt}
\end{equation}
Here, the number of solutions $t$ must be known to obtain $L_{\mathrm{opt}}$, but in practice, it cannot be obtained in advance.
To cope with this problem, Boyer \etal proposed the BBHT algorithm.
As summarized in Algorithm~\ref{alg:BBHT}, it searches for an appropriate $L$ until the desired solution is found.
This repetition achieves the query complexity of $O\qty(\sqrt{ N / t})$ even if the number of solutions is unknown.
To be more specific, $L$ is a random variable drawn from a uniform distribution $[0, k)$, where the parameter $k$ is increased by $k_{i+1} = \min\qty{\lambda k_i, \sqrt{2^{q_1}}}$ for each iteration.
The parameter $\lambda$ is a constant value related to the increase rate of $k$, typically ranging between $1< \lambda < 4/3$ \cite{BOYER1998TIGHT}.
\begin{algorithm}[tb]
    \caption{BBHT Algorithm~\cite{BOYER1998TIGHT}.\label{alg:BBHT}}
    \begin{algorithmic}[1]
        \renewcommand{\algorithmicrequire}{\textbf{Input:}}
        \renewcommand{\algorithmicensure}{\textbf{Output:}}
        \REQUIRE $\lambda > 1$
        \ENSURE $\mathbf{x}$
        
        \STATE {Set $k_0 = 1$ and $i = 0$}.
        \WHILE{Optimal solution not found}
        \STATE\hspace{\algorithmicindent}{Randomly select the rotation count $L_i$ from the set $\{0, 1, ..., \lceil k_i-1 \rceil$\}}.
        \STATE\hspace{\algorithmicindent}{Evaluate $\mathbf{G}^{L_i} \mathbf{A} \Ket{0}_{q_{1}}$ to obtain $\mathbf{x}$}.
        \STATE\hspace{\algorithmicindent}{$k_{i+1}=\min{\{\lambda k_i,\sqrt{2^{q_{1}}}}\}$}.
        \ENDWHILE
    \end{algorithmic}
\end{algorithm}

\subsection{GAS Algorithm \cite{DURR1999QUANTUM,BULGER2003IMPLEMENTING,BARITOMPA2005GROVER,GILLIAM2021GROVER}\label{SUBSEC:GAS}}
GAS is an algorithm that extends the BBHT algorithm to obtain the optimal solution for a minimization problem
\begin{equation}
    \begin{aligned}
        \min_{\mathbf{x}} \quad & E(\mathbf{x}) \\
        \textrm{s.t.} \quad & x_i \in \mathbb{F}_2 = \{0, 1\} ~~~~ (i = 0, \cdots, q_1-1).
    \end{aligned}
\end{equation}
\begin{algorithm}[tb]
    \caption{Conventional GAS~\cite{BARITOMPA2005GROVER,GILLIAM2021GROVER}.\label{alg:conv_GAS}}
    \begin{algorithmic}[1]
        \renewcommand{\algorithmicrequire}{\textbf{Input:}}
        \renewcommand{\algorithmicensure}{\textbf{Output:}}
        \REQUIRE $E:\mathbb{F}_2^{q_{1}}\rightarrow\mathbb{Z}, \lambda=1.34$
        \ENSURE $\mathbf{x}$
        \STATE {Uniformly sample $\mathbf{x}_0 \in \mathbb{F}_2^{q_{1}} $ and set $y_0=E(\mathbf{x}_0)$}.

        \STATE {Set $k = 1$ and $i = 0$}.

        \REPEAT
        \STATE\hspace{\algorithmicindent}{Randomly select the rotation count $L_i$ from the set $\{0, 1, ..., \lceil k-1 \rceil$\}}.
        \STATE\hspace{\algorithmicindent}{Evaluate $\mathbf{G}^{L_i} \mathbf{A}_{y_i} \Ket{0}_{q_{1} + q_{2}}$ to obtain $\mathbf{x}$ and $y$}.
        \hspace{\algorithmicindent}\IF{$y<y_i$}
        \STATE\hspace{\algorithmicindent}{$\mathbf{x}_{i+1}=\mathbf{x}, y_{i+1}=y,$ and $k=1$}.
        \hspace{\algorithmicindent}\ELSE{\STATE\hspace{\algorithmicindent}{$\mathbf{x}_{i+1}=\mathbf{x}_i, y_{i+1}=y_i,$ and $k=\min{\{\lambda k,\sqrt{2^{q_{1}}}}\}$}}.
        \ENDIF
        \STATE{$i=i+1$}.
        \UNTIL{a termination condition is met}.
    \end{algorithmic}
\end{algorithm}
Here, $\mathbf{x} = (x_0, \cdots, x_{q_1 - 1})$ denotes binary variables and $E(\mathbf{x}) \in \mathbb{Z}$ denotes an arbitrary objective function.
As with the BBHT algorithm, the query complexity is $O\qty(\sqrt{N / t})$, where $N$ is the search space size, i.e., $N = 2^{q_1}$.
As summarized in Algorithm~\ref{alg:conv_GAS}, first, the initial value of the threshold $y_0 = E(\mathbf{x}_0)$ is determined by a random solution $\mathbf{x}_0$.
Next, the BBHT algorithm is used to search for a better solution $\mathbf{x}$ such that the value of the objective function is less than the current threshold $y_{i}$, and the threshold is updated by $y_{i+1} = E(\mathbf{x})$.
Thereby, the optimal solution is obtained through multiple trials.
The optimal value of the parameter $\lambda$ that minimizes the query complexity is claimed to be 1.34 \cite{BARITOMPA2005GROVER}.

\begin{figure}[tb]
    \centering  \includegraphics[keepaspectratio,scale=0.7]{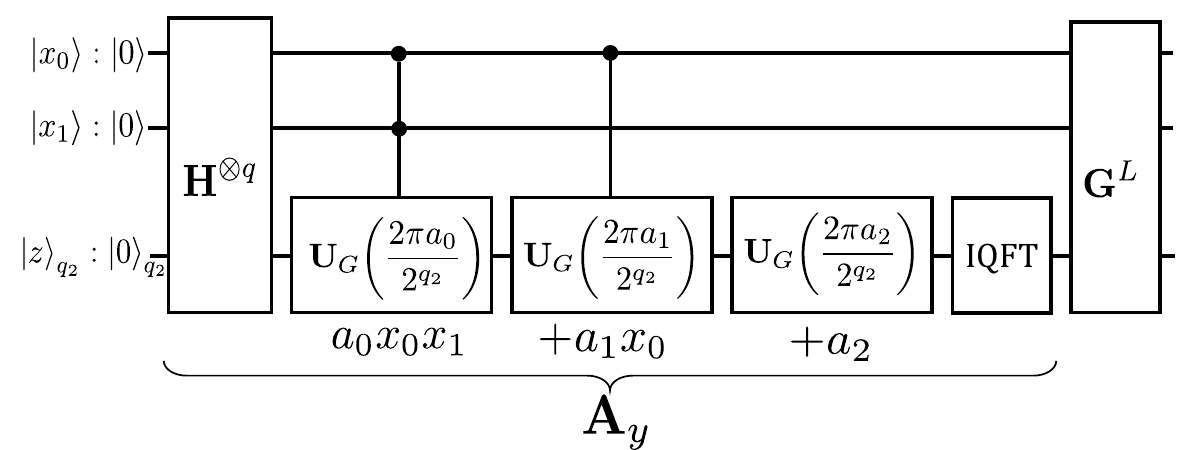}
    \begingroup
    \newlength{\xfigwd}
    \setlength{\xfigwd}{\textwidth}
    \caption{Quantum circuit corresponding to $E(\mathbf{x}) - y = a_0x_0x_1 + a_1x_0 + a_2$.}
    \label{fig:circuit1}
    \endgroup
\end{figure}
Most conventional studies have assumed the existence of a black-box quantum oracle $\mathbf{O}$ that can identify the states of interest.
The breakthrough was Gilliam's method \cite{GILLIAM2021GROVER}.
In \cite{GILLIAM2021GROVER}, Gilliam \etal proposed an efficient construction method for a quantum circuit that corresponds to an objective function of binary optimization.
Aided by two's complement representation of the objective function value, the quantum oracle $\mathbf{O}$ can be constructed easily.
As an explicit example, Fig.~\ref{fig:circuit1} shows a quantum circuit corresponding to the objective function $E(\mathbf{x}) - y = a_0x_0x_1 + a_1x_0 + a_2$, which has quantum gates that correspond to the polynomial terms $a_0x_0x_1$, $a_1x_0$, and $a_2$.
Note that $a_2$ is a constant that includes the constant term of $E(\mathbf{x})$ and $-y$.

To investigate the details of Gilliam's construction method, $q = q_{1} + q_{2}$ qubits are required. $q_{1}$ denotes the number of binary variables $(x_0, \cdots, x_{q_1})$, and $q_{2}$ denotes the number of qubits for encoding the objective function $E(\mathbf{x})$.
Here, $q_{2}$ satisfies
\begin{equation}
    -2^{q_{2} - 1} \leq \min_{\mathbf{x}} E(\mathbf{x}) \leq \max_{\mathbf{x}} E(\mathbf{x}) < 2^{q_{2} - 1}
\end{equation}
due to the two's complement representation.
Next, we describe the procedure for constructing a quantum circuit \cite{GILLIAM2021GROVER}.
\begin{enumerate}
    \item Hadamard gates $\mathbf{H}^{\otimes q}$ act on the initial state $\Ket{0}_{q}$ to create a uniform superposition state, yielding the transition of \cite{GILLIAM2021GROVER}
    \begin{equation}
    \begin{split}
        \Ket{0}_{q} \xrightarrow{\mathbf{H}^{\otimes q}} &\frac{1}{\sqrt{2^{q}}} \sum_{i=0}^{2^{q}-1} \Ket{i}_{q}\\
        = &\frac{1}{\sqrt{2^{q_1 + q_2}}} \sum_{\mathbf{x} \in \mathbb{F}_2^{q_1}} \sum_{i=0}^{2^{q_2}-1} \Ket{\mathbf{x}}_{q_1} \Ket{i}_{q_2}.
    \end{split}
    \label{eq:AyH}
    \end{equation}
    Using the tensor product $\mathbf{H} \otimes \mathbf{H}$, the Hadamard gates $\mathbf{H}^{\otimes q}$ are defined as
    \begin{equation}
        \mathbf{H}^{\otimes q} = \underbrace{\mathbf{H} \otimes \mathbf{H} \otimes \cdots \otimes \mathbf{H}}_{q}
    \end{equation}
    and
    \begin{equation}
        \mathbf{H} = \frac{1}{\sqrt{2}}\mqty[1 & 1 \\ 1 & -1].
    \end{equation}
    \item Each term of $E(\mathbf{x}) - y$ corresponds to a unitary gate $\mathbf{U}_{G}(\theta)$. It acts on the lower $q_2$ qubits to rotate the phase of quantum states. Let $a$ be a coefficient of the term. Then, the phase shift is defined as $\theta = \frac{2 \pi a}{2^{q_2}}$. That is, the phase advance represents addition, and the phase delay represents subtraction. The unitary gate is defined by
    \begin{equation}
        \mathbf{U}_{G}(\theta) = \underbrace{\mathbf{R}(2^{q_2 - 1}\theta) \otimes \cdots \otimes \mathbf{R}(2^{0}\theta)}_{q_2}
    \end{equation}
    and the phase gate
    \begin{equation}
        \mathbf{R}(\theta) = \mqty[1 & 0 \\ 0 & e^{j\theta}].
    \end{equation}
    For a state $\mathbf{x}_0 \in \mathbb{F}_2^{q_1}$, $\theta' = 2 \pi \qty(E(\mathbf{x}_0) - y) / 2^{q_2}$ and $\mathbf{U}_{G}(\theta')$ yield the transition of \cite{GILLIAM2021GROVER}
    \begin{equation}
        \frac{1}{\sqrt{2^{q_2}}} \sum_{i=0}^{2^{q_2}-1} \Ket{i}_{q_2} \xrightarrow{\mathbf{U}_{G}\qty(\theta')}
        \frac{1}{\sqrt{2^{q_2}}}\sum_{i=0}^{2^{q_2} - 1}e^{ji\theta'} \Ket{i}_{q_2}.
        \label{eq:AyUG}
    \end{equation}
    The interaction between binary variables is expressed by controlled $\mathbf{U}_{G}(\theta)$, as shown in Fig.~\ref{fig:circuit1}.
    \item Applying the inverse quantum Fourier transform (IQFT) \cite{shor1997polynomialtime} to the lower $q_{2}$ qubits yields the transition of \cite{GILLIAM2021GROVER}
    \begin{equation}
        \frac{1}{\sqrt{2^{q_2}}}\sum_{i=0}^{2^{q_2} - 1}e^{ji\theta'} \Ket{i}_{q_2}  \xrightarrow{\mathrm{IQFT}} \frac{1}{\sqrt{2^{q_2}}}\Ket{E(\mathbf{x}_0) - y}_{q_2}.
    \end{equation}
    In summary, steps 1--3 are referred to as a state preparation operator $\mathbf{A}_{y}$. It encodes the value $E(\mathbf{x}) - y$ into $q_2$ qubits for all the states. That is, it satisfies \cite{GILLIAM2021GROVER}
    \begin{equation}
        \mathbf{A}_{y} \Ket{0}_{q} = \frac{1}{\sqrt{2^{q_{1}}}}\sum_{\mathbf{x} \in \mathbb{F}_2^{q_1}} \Ket{\mathbf{x}}_{q_{1}}\Ket{E(\mathbf{x})-y}_{q_{2}}.
    \end{equation}
    \item Finally, to amplify the states of interest, the Grover operator $\mathbf{G} = \mathbf{A}_{y}\mathbf{D}\mathbf{A}_{y}^H\mathbf{O}$ acts $L$ times, where $\mathbf{D}$ denotes the Grover diffusion operator~\cite{grover1996fast}. The oracle $\mathbf{O}$ identifies the states of interest that satisfy $E(\mathbf{x}) - y < 0 \Leftrightarrow E(\mathbf{x}) < y$. Since we use the two's complement representation, such states can be identified by a single Pauli-Z gate
    \begin{equation}
        \mathbf{Z} = \mqty[1 & 0 \\ 0 & -1],
    \end{equation}
    which is acted on the beginning of $q_2$ qubits.
\end{enumerate}

Note that an open-source implementation of the GAS algorithm is available from IBM Qiskit~\cite{qiskit}.

\section{Proposed Quantum Search Algorithm}
In this section, we propose a novel objective function that attributes the problem of finding a binary constant weight code to an optimization problem. The problem is defined as finding a code such that the code length $n$, weight $w$, number of codewords $M$, and minimum distance $d$ are each equal to the given values $(n, w, M, d)$, which is hereinafter referred to as condition $(n, w, M, d)$.
Additionally, focusing on the inherent properties of the problem, we derive an upper bound on the minimum value of the objective function and a lower bound on the number of solutions.
Using these bounds, we modify GAS and reduce the number of qubits and query complexity.

\subsection{Proposed Formulation}
To obtain the optimal solution for a binary constant weight code $A(n,d,w)$, we newly formulate the search problem as a QUBO problem, which is supported by GAS.
A challenging task here is that the objective function has to incorporate the Hamming distance between codewords.
In general, the Hamming distance $d(\c_m, \c_{m'})$ between codewords $\c_m$ and $\c_{m'}$ is expressed using the XOR operation $\oplus$, which is not a QUBO function.
As long as the Hamming weight is constant, the XOR operation can be expressed by an inner product $\expval{\c_m, \c_{m'}} = \sum_{i=1}^nc_{m, i} c_{m', i}$, where $c_{m, i}$ is the $i$-th element of $\c_m$.
\begin{thm}
If a code has a constant weight $w$, then, for any pair of codewords $(\c_{m}, \c_{m'})$, we have
\begin{equation}
    \label{eq:hamming_distance}
    d(\c_{m}, \c_{m'}) =  2\qty(w - \expval{\c_m, \c_m'}).
\end{equation}
\end{thm}
\begin{proof}
For any pair of $i$-th elements of codewords  $(c_{m, i}, c_{m', i}) = (0, 0), (0, 1), (1, 0), (1, 1)$, the following equation holds:
\begin{equation}
    \label{eq:identity}
    c_{m, i}\oplus c_{m', i} = \qty(c_{m, i} + c_{m', i}) - 2 c_{m, i}c_{m', i}.
\end{equation}
Since weights of codewords are constant, $w$, by summing \eqref{eq:identity} over index $i=1,\cdots, n$, we have
\begin{align}
    \sum_{i=1}^n\qty(c_{m, i}\oplus c_{m', i}) =& d(\c_m, \c_{m'}) \nonumber\\
    =& \sum_{i=1}^n \qty(c_{m, i} + c_{m', i}) - 2 \sum_{i=1}^nc_{m, i} c_{m', i} \nonumber\\
    =& 2\qty(w - \expval{\c_m, \c_{m'}})
\end{align}
\end{proof}
Theorem~1 allows us to transform the minimum distance of a code $d(\C)$ into
\begin{align}
    \label{eq:minimum_hamming_dsitance}
    d(\C) &= \min_{0 \leq m < m' < M} d(\c_m, \c_{m'}) \nonumber\\
    &= \min_{0 \leq m < m' < M}  2\qty(w -  \expval{\c_m, \c_{m'}}) \nonumber\\
    &= 2\qty(w - \max_{0 \leq m < m' < M} \expval{\c_m, \c_{m'}}).
\end{align}
Here, the following Theorem holds.
\begin{thm}
The minimum distance of a code $\C^*$ is at least $d$, where we have
\begin{align}
    \C^* &= \argmin_{C} s(\C), \label{eq:minimize} \\
    s(\C) &= \sum_{0 \leq m < m' < M} \expval{\c_m, \c_{m'}}^l,~~\mathrm{and} \\
    l &= \left\lfloor\frac{\log \binom{M}{2}}{\log \qty(1 + \frac{2}{2w - d})} + 1\right\rfloor,
\end{align}
which is valid for the $d \neq 2w$ case.\footnote{If $d = 2w$, according to \eqref{eq:minimum_hamming_dsitance}, an inner product of any pair of codewords is zero. Thus, the solution is obvious, and the maximum number of codewords is $\left\lfloor\frac{n}{w}\right\rfloor$.}
\end{thm}
\begin{proof}
From \eqref{eq:hamming_distance}, for any pair $(m, m')$ of a code such that the minimum distance is at least $d$, the following inequality holds:
\begin{equation}
    \label{eq:d_condition}
    \expval{\c_m, \c_{m'}} \leq w - \frac{d}{2}.
\end{equation}
From this relationship, the upper bound on $s(\C)$ is given by
\begin{equation}
    \label{eq:lower_bound_obj1}
    s(\C) \leq \overline{s}(\C) = \binom{M}{2}\qty(w - \frac{d}{2})^l,
\end{equation}
where $l$ is an arbitrary coefficient.
Conversely, for a code such that the minimum distance is less than $d$, there is at least one pair $(m, m')$ that satisfies
\begin{align}
    \expval{\c_m, \c_{m'}} \geq w - \frac{d-2}{2} = w - \frac{d}{2} + 1.
\end{align}
From this relationship, the lower bound on $s(\C)$ is given by
\begin{equation}
    s(\C) \geq \underline{s}(\C) = \qty(w - \frac{d}{2} + 1)^l.
\end{equation}
By comparing the obtained upper and lower bounds, if $l$ is sufficiently large, the following inequality holds:
\begin{align}
    \binom{M}{2}\qty(w - \frac{d}{2})^l &< \qty(w - \frac{d}{2} + 1)^l.
\end{align}
Thus, if $l$ is sufficiently large, \eqref{eq:minimize} is minimized, and the minimum distance of the obtained code becomes at least $d$. Taking a logarithm of both sides, the condition for $l$ is given by
\begin{equation}
    l > \frac{\log\binom{M}{2}}{\log\qty(1 + \frac{2}{2w - d})}.
\end{equation}
\end{proof}
Based on the above relationships, we formulate the problem of finding a binary constant weight code as a QUBO problem, while maximally reducing the number of binary variables.
We propose an objective function of
\begin{equation}
    \label{eq:obj_func1}
     E(\mathbf{x}) = f(\mathbf{x}) + \rho \cdot g(\mathbf{x}),
\end{equation}
where
\begin{align}
    f(\mathbf{x}) = \sum_{0 \leq r < r' < \binom{n}{w}} \expval{\mathbf{p}_r, \mathbf{p}_{r'}}^l \cdot x_r x_{r'}
\end{align}
and
\begin{align}
    g(\mathbf{x}) = \qty(\sum_{r=0}^{\binom{n}{w}-1}x_r - M)^2.
\end{align}
Here, for index $r = 0, ~ \cdots, ~ \binom{n}{w} - 1$, $\mathbf{p}_r$ denotes the $(r+1)$-th row vector of the combinatorial matrix $\mathbf{P}(n, w)$, which represents $\binom{n}{w}$ combinations~\cite{frenger1999parallel}
\begin{equation}
    \mathbf{P}(n, w) = \mqty[\mathbf{1} & \mathbf{P}(n-1, w-1) \\ \mathbf{0} & \mathbf{P}(n-1, w)] \in \mathbb{F}_2^{\binom{n}{w} \times n}.
    \label{eq:Pnw}
\end{equation}
In \eqref{eq:Pnw}, $\mathbf{1}$ denotes a one column vector of length $\binom{n-1}{w-1}$ and $\mathbf{0}$ denotes a zero column vector of length $\binom{n-1}{w}$.
In \eqref{eq:obj_func1}, the binary variable $x_r$ indicates whether the $(r+1)$-th row of the combinatorial matrix is used as a codeword or not. That is, the relationship between $\C$ and $x_r$ is expressed by $\C = \qty{\mathbf{p}_r \middle | x_{r} = 1}$.
The penalty function $g(\mathbf{x})$ ensures that only $M$ from the $\binom{n}{w}$ combinations are used and has priority over the cost function $f(\mathbf{x})$.
Thus, the penalty coefficient $\rho$ can be defined as
\begin{align}
    \rho = \overline{f} + 1 = \binom{M}{2}(w - 1)^l + 1,
\end{align}
where $\overline{f}$ denotes the upper bound of $f$.

\subsection{Reduction of Qubits}
Since the number of qubits required to construct a quantum circuit determines feasibility, it should be maximally reduced.
Here, we modify the objective function \eqref{eq:obj_func1} so that a single codeword is always selected, leading to a reduction in the number of qubits. 
\begin{thm}
\label{thm:uniformity}
If a code that satisfies the condition $(n, w, M, d)$ exits,
then for every possible codeword $\c$ in the combinatorial matrix, at least one code exists that satisfies the
condition $(n, w, M, d)$ such that the codeword $\c$ is included.
\end{thm}
\begin{proof}
When we consider a code as a matrix where each row vector is a codeword,
the Hamming distance between codewords does not change even if we swap any two columns of the matrix.
Therefore, by swapping the columns of any code that satisfies the condition $(n, w, M, d)$,
we can construct a code that contains any specific codeword and satisfies the condition $(n, w, M, d)$.
\end{proof}
Let $\mathbf{P}'(n,w)$ be a matrix generated by removing $\mathbf{P}(n,w)$'s row vectors, the Hamming distance of which from a codeword $\mathbf{p}_0$ is less than $d$, where $\mathbf{p}_0$ is the first row vector of $\mathbf{P}(n,w)$ \eqref{eq:Pnw}.
The objective function with reduced binary variables can be defined as
\begin{equation}
    E'(\mathbf{x}) = f'(\mathbf{x}) + \rho' \cdot g'(\mathbf{x}),
    \label{eq:obj_func2}
\end{equation}
where
\begin{align}
    \label{eq:f'}
    f'(\mathbf{x}) &= \sum_{0 \leq r < r' < q_{1}} \expval{\mathbf{p'}_r, \mathbf{p'}_{r'}}^l \cdot x_r x_{r'},\\
    \label{eq:g'}
    g'(\mathbf{x}) &= \qty(\sum_{r=0}^{q_{1}-1}x_r - (M - 1))^2,
\end{align}
and $\mathbf{p}'_r$ is the $(r+1)$-th row vector of $\mathbf{P}'(n, w)$.
The number of binary variables $q_{1}$ is equal to the number of rows in $\mathbf{P}'(n, w)$ and is expressed as
\begin{equation}
    \label{eq:bit_num1}
    q_{1} = \sum_{i=0}^{w - \frac{d}{2}} \binom{w}{i}\binom{n-w}{w-i}.
\end{equation}
The penalty coefficient $\rho'$ can be defined as
\begin{align}
    \label{eq:rho'}
    \rho' = \overline{f}' + 1 = \binom{q_{1}}{2}(w - 1)^l + 1
\end{align}
where $\overline{f}'$ denotes the upper bound of $f'$.
Let $E'_{\mathrm{max}}$ be the maximum of the objective function value, i.e., $E'_{\mathrm{max}} = \max_{\mathbf{x}}E'(\mathbf{x})$.
It is upper bounded by
\begin{align}
    &E'_{\mathrm{max}} \leq \overline{E}'_{\mathrm{max}} = \overline{f}' + \rho'\overline{g}' \nonumber\\
    =&\begin{cases}
        \binom{q_{1}}{2}(w-1)^l + \rho' \qty(q_{1} - M + 1)^2 & (2(M-1) < q_{1})  \\
        \binom{q_{1}}{2}(w-1)^l + \rho' (M-1)^2 & (2(M-1) \geq q_{1})
    \end{cases} \nonumber\\
    =&\begin{cases}
        O\qty(q_1^2w^l(q_{1} - M)^2) & (2(M-1) < q_{1}) \\
        O\qty(q_1^2w^lM^2) & (2(M-1) \geq q_{1}).
    \end{cases}
\end{align}
Then, the number of qubits required to encode the objective function value is expressed by
\begin{align}
    &q_{2}' = \lceil\log_2(\overline{E}'_{\mathrm{max}})\rceil + 1 \nonumber\\
    =& \begin{cases}
        O\qty(\log(q_1) + l\log(w)) & (2(M-1) < q_{1})  \\
        O\qty(\log(q_1) + l\log(w) + \log(M)) & (2(M-1) \geq q_{1}).
    \end{cases}
    \label{eq:bit_num2_order}
\end{align}
Assuming $w \neq 2d$, the coefficient $l$ is upper bounded by
\begin{align}
    \label{eq:l_upper}
    l &= \left\lfloor\frac{\log \binom{M}{2}}{\log \qty(1 + \frac{2}{2w - d})} + 1\right\rfloor \leq \left\lfloor\frac{\log \binom{M}{2}}{2} + 1\right\rfloor \nonumber\\
    &= O(\log(M))
\end{align}
Using \eqref{eq:l_upper}, $q_2'$ of \eqref{eq:bit_num2_order} can be simplified to
\begin{equation}
    q_{2}' = O\qty(\log(q_1) + \log(M)\log(w)).
\end{equation}
The total number of qubits is
\begin{equation}
    \label{eq:bit_num_all}
    q_{1} + q_{2}' = O\qty(q_{1} + \log(M)\log(w)).
\end{equation}
The query complexity in the quantum domain is expressed by
\begin{equation}
    O\qty(\sqrt{\frac{2^{q_{1}}}{t}}),
\end{equation}
where $t$ is the exact number of solutions and is analyzed later.

\subsection{Further Speedup for Grover Adaptive Search}
\subsubsection{Novel Initial Threshold}
Let $E'_{\mathrm{min}}$ be the minimum of the objective function value, i.e., $E'_{\mathrm{min}} = \min_{\mathbf{x}}E'(\mathbf{x})$.
Similar to \eqref{eq:lower_bound_obj1}, its upper bound can be derived as
\begin{equation}
    \label{eq:lower_bound_obj2}
    E'_{\mathrm{min}} < \overline{E}'_{\mathrm{min}} + 1 = \binom{M-1}{2}\qty(w - \frac{d}{2})^l + 1.
\end{equation}
The objective function value $E'(\mathbf{x})$ may become greater than $\overline{E}'_{\mathrm{min}}$.
By contrast, if it is smaller than $\overline{E}'_{\mathrm{min}}$, the obtained code has at least the minimum distance $d$.
Therefore, the desired code can be obtained without updating the threshold of GAS.
That is, we set the initial value of the threshold to
\begin{equation}
    y_0 = \overline{E}'_{\mathrm{min}} + 1. \label{eq:y0Emin}
\end{equation}

Additionally, if the objective function value is greater than $\overline{E}'_{\mathrm{min}}$, the obtained code is not good regarding the minimum distance.
This fact helps to reduce the number of qubits $q_2$.
Specifically, the penalty coefficient $\rho'$ is replaced by
\begin{equation}
    \rho'' = \overline{E}'_{\mathrm{min}} + 1. \label{eq:rho''}
\end{equation}
and we have an updated objective function
\begin{equation}
    E''(\mathbf{x}) = f'(\mathbf{x}) + \rho'' \cdot g'(\mathbf{x})
    \label{eq:obj_func3}
\end{equation}
where $f'$ and $g'$ are defiend by \eqref{eq:f'} and \eqref{eq:g'}, respectively. Then, the objective function is upper bounded by
\begin{align}
    \label{eq:obj_func3_upper}
    &E''_{\mathrm{max}} \leq \overline{E}''_{\mathrm{max}} \nonumber\\
    =&\begin{cases}
        O\qty(q_1^2w^l + M^2(w - d/2)^l(q_{1} - M)^2) & (2(M-1) < q_{1}) \\
        O\qty(q_1^2w^l + M^4(w - d/2)^l) & (2(M-1) \geq q_{1}).
    \end{cases}
\end{align}
As with the $E_{\mathrm{max}}'$ case, the number of qubits required to encode the objective function value can be simplified to
\begin{align}
    \label{eq:bit_num3}
    q_{2}'' &= \lceil\log_2(\overline{E}''_{\mathrm{max}})\rceil + 1 \nonumber\\
    &= O\qty(\log(q_1) + \log(M)\log(w)).
\end{align}
The number of qubits that can be reduced is bounded by
\begin{align}
    &\lceil\log_2( \overline{E}'(\mathbf{x}) )\rceil - \lceil\log_2( \overline{E}''(\mathbf{x}) )\rceil \nonumber\\
    > & \log_2\qty(\frac{\overline{E}'(\mathbf{x})}{\overline{E}''(\mathbf{x})}) - 2 =  \log_2\qty(\frac{\overline{f}' + \rho'\overline{g}' - 1}{\overline{f}' + \rho''\overline{g}' - 1}) - 2 \nonumber\\
    > & \log_2\qty(\frac{\overline{f}' + \rho'\overline{g}'}{\overline{f}' + \rho''\overline{g}'}) - 2 = \log_2\qty(\frac{1 + \overline{g}'}{1 + \frac{\rho''}{\rho'}\overline{g}'}) - 2 \nonumber\\
    = & \log_2\qty(\frac{1 + \overline{g}'}{1 + \frac{(M-1)(M-2)(2w - d)^l + 4}{q_1(q_1 - 1)(2w-2)^l + 4}\overline{g}'}) - 2,
\end{align}
which is obviously positive since $\rho' > \rho''$ when $q_1 > M$ and $d > 2$.

Note that, for constructing a code that satisfies the condition $(n, w, M, d)$, it is sufficient to obtain just one code whose objective function is smaller than $\overline{E}'_{\mathrm{min}}$.

\subsubsection{Novel Limit on the Number of Grover Operators}
The BBHT search and GAS algorithms try to find the appropriate number of Grover operators, which increases the query complexity compared to the case where the number of solutions is previously known.
Although, in general, the exact number of solutions $t$ is unknown, the possible range of $t$ may be obtained from the inherent structure of the problem.
Our proposed algorithm exploits this information.
\begin{thm}
\label{THM:LOWER_BOUNDS_SOLUTION}
If $w \leq d$ and $A(n-1,  d, w) < M-1$, then the exact number of solutions $t$ is lower bounded by
\begin{align}
    t &\geq \underline{t} = \left\{
    \begin{array}{ll}
        w! & (w-\frac{d}{2}=1) \\
        \displaystyle{\min_{2 \leq i \leq w - d/2}} \ \binom{w}{i} & (w-\frac{d}{2} \geq 2).
    \end{array} \label{eq:tlower}
    \right.
\end{align}
\end{thm}
\noindent The proof is given in the Appendix~\ref{GILLIAM2021GROVER}.

We improve the quantum search algorithm using the obtained lower bound $\underline{t}$ of \eqref{eq:tlower}.
Traditionally, the number of Grover operators $L$ is uniformly drawn from the semi-open interval $[0, k)$, where $k$ is a parameter that increases in each iteration of GAS.
The probability of success $P_k(t)$ is known as~\cite{BOYER1998TIGHT}
\begin{equation}
    \label{eq:succeed_prob}
    P_k(t) = \frac{1}{2} - \frac{\sin(4k\theta)}{4k\sin(2\theta)}
\end{equation}
where
\begin{equation}
    \theta = \arcsin\qty(\sqrt{\frac{t}{2^{q_{1}}}}).
\end{equation}
Then, the optimal value $k_{\mathrm{opt}}$ is
\begin{align}
    \label{eq:k_opt}
    k_{\mathrm{opt}} &= \argmin_{k}\  \frac{k}{P_k(t)} \nonumber \\
    &= \argmin_{k}\  \frac{4k^2\sin(2 \theta)}{2k\sin(2\theta) - \sin(4k\theta)}
\end{align}
and is upper bounded by
\begin{equation}
    k_{\mathrm{opt}} \leq \overline{k}_{\mathrm{opt}} = \argmin_{k}\
    \frac{k}{P_k(\underline{t})}.
    \label{eq:koptup}
\end{equation}
\begin{thm}
\label{THM:UPPER_BOUND_K}
If the ratio $t/2^{q_1}$ is sufficiently small, then
\begin{equation}
    \overline{k}_{\mathrm{opt}} \in \qty[1, \left\lceil\frac{1+\sqrt{2}}{2} \sqrt{\frac{2^{q_1}}{\underline{t}}}\right\rceil].
\end{equation}
\end{thm}
\noindent The proof is given in the Appendix~\ref{app:thm_k}.
In general, $t/2^{q_1}$ is sufficiently small that it can be approximated as $(t/2^{q_1})^2 \sim 0$.
The search range of $\overline{k}_{\mathrm{opt}}$ can be limited within the closed interval using Theorem~\ref{THM:UPPER_BOUND_K}.
Then, we can search for $\overline{k}_{\mathrm{opt}}$ with negligible complexity; however, $\overline{k}_{\mathrm{opt}}$ itself is not expressed in a closed form.
For example, within the closed interval, root-finding algorithm such as the bisection method~\cite{Burden1989} or Newton-Raphson method~\cite{galantai2000theory} is used for the first-order derivative of $k / P_k(\underline{t})$, all the $k$ that take the extreme values are listed, and $\overline{k}_{\mathrm{opt}}$ is one of the listed values or one of both ends of the closed interval. Note that, the lower bound obtained in Theorem~\ref{THM:LOWER_BOUNDS_SOLUTION} requires assumption $w \leq d$ and $A(n-1,  d, w) < M-1$. Thus, this extra technique cannot be applied to all possible parameters of the problem. Since there has not been much research on the number of solutions for constant weight codes, more general and stronger bounds may be found in the future.
Moreover, Theorem 4 also holds for $k_\mathrm{opt}$ by replacing $\underline{t}$ with $t$. This result suggests that the value of $k$ need not be larger than $(1 + \sqrt{2})/2 \cdot \sqrt{2^{q_1}} \sim 1.207 \cdot \sqrt{2^{q_1}}$, even if the number of solutions is completely unknown. A similar technique was proposed in \cite{BARITOMPA2005GROVER} to set the upper bound of $k$ to $\sqrt{2^{q_1}}$. Theorem~\ref{THM:UPPER_BOUND_K} provides a rough basis for this technique.

The query complexity can be reduced by setting $\overline{k}_{\mathrm{opt}}$ as the upper bound on $k$.
Note that, although this paper only focuses on the lower bound, an upper bound on $t$ may be obtained for some problems, and a similar approach could speedup the quantum search algorithm more.

\begin{algorithm}[tb]
    \caption{Proposed GAS.\label{alg:proposed_GAS}}
    \begin{algorithmic}[1]
        \renewcommand{\algorithmicrequire}{\textbf{Input:}}
        \renewcommand{\algorithmicensure}{\textbf{Output:}}
        \REQUIRE $E'':\mathbb{F}_2^n\rightarrow\mathbb{Z}, \overline{E}_{\mathrm{min}}, \overline{k}_{\mathrm{opt}}, \lambda = 1.44$
        \ENSURE $\mathbf{x}$
        \STATE {Set $y_0 = \overline{E}_{\mathrm{min}}$, $k = 1$, and $i = 0$}.

        \REPEAT
        \STATE\hspace{\algorithmicindent}{Randomly select the rotation count $L_i$ from the set $\{0, ..., \lceil k-1 \rceil$\}}.
        \STATE\hspace{\algorithmicindent}{Evaluate $\mathbf{G}^{L_i} \mathbf{A}_{y_i} \Ket{0}_{n+m}$ to obtain $\mathbf{x}$ and $y$}.
        \hspace{\algorithmicindent}\IF{$y<y_i$}
        \STATE\hspace{\algorithmicindent}{$\mathbf{x}_{i+1}=\mathbf{x}, y_{i+1}=y,$ and $k = 1$}.
        \hspace{\algorithmicindent}\ELSE{\STATE\hspace{\algorithmicindent}{$\mathbf{x}_{i+1}=\mathbf{x}_i, y_{i+1}=y_i,$ and $k=\min{\{\lambda k, \overline{k}_{\mathrm{opt}} }\}$}}.
        \ENDIF
        \STATE{$i=i+1$}.
        \UNTIL{a termination condition is met}.
    \end{algorithmic}
\end{algorithm}
Overall, the proposed GAS is summarized in Algorithm~\ref{alg:proposed_GAS}.
The conventional GAS is initiated with a random solution $\mathbf{x}_0$ and a random threshold $y_0 = E''(\mathbf{x}_0)$.
In the proposed GAS, a strict threshold is set from the beginning, using the upper bound on the minimum of the objective function value, $\overline{E}'_{\mathrm{min}} + 1$ of \eqref{eq:y0Emin}.
Additionally, instead of the conventional upper bound $k \leq \sqrt{2^{q_1}}$, we use the stricter bound, $\overline{k}_{\mathrm{opt}}$ of \eqref{eq:koptup}.
These modifications further accelerate the efficient GAS algorithm.
Note that, in \cite{BARITOMPA2005GROVER}, the probability of a success is maximized with $\lambda = 1.44$ if the number of solutions is sufficiently smaller than the search space size (about 0.2$\%$ in \cite{BARITOMPA2005GROVER}).
In general, if the initial threshold is as strict as $\overline{E}'_{\mathrm{min}}$, the number of solutions becomes much smaller than the search space size.
Thus, we set $\lambda = 1.44$ in Algorithm~\ref{alg:proposed_GAS} following \cite{BARITOMPA2005GROVER}.

\subsection{Number of Quantum Gates}
We analyzed the number of quantum gates that determines the size of a quantum circuit, which closely relates to the feasibility of the proposed algorithm.
In the GAS circuit, the number of gates required for $\mathbf{A}_{y_i}$ has a dominant impact, which corresponds to the state preparation operator.
Thus, we only focus on $\mathbf{A}_{y_i}$ later.

As described in the first step of Section~\ref{SUBSEC:GAS}, the Hadamard gate $\mathbf{H}$ acts once on every qubit.
Thus, using \eqref{eq:bit_num_all}, the number of $\mathbf{H}$ gates involved in \eqref{eq:AyH} is
\begin{equation}
    G_\textrm{H} = q_1 + q_2 =  O\qty(q_{1} + \log(M)\log(w)).
\end{equation}
In the second step, the phase gate $\mathbf{R}$ is used to represent the coefficients in the objective function.
Using \eqref{eq:bit_num2_order}, the number of $\mathbf{R}$ gates involved in \eqref{eq:AyUG} is
\begin{equation}
    G_\textrm{R} = q_2 = O\qty(\log(q_1) + \log(M)\log(w)).
\end{equation}
The number of 1-controlled phase (1-CR) gates is the same as that of first-order terms in the objective function and is expressed as
\begin{equation}
    G_\textrm{1-CR} = q_1 \cdot q_2 = O\qty(q_1\qty(\log(q_1) + \log(M)\log(w))).
\end{equation}
Similarly, the number of 2-controlled phase (2-CR) gates is the same as that of second-order terms and is expressed by
\begin{equation}
    G_\textrm{2-CR} = \frac{q_1(q_1 - 1)}{2} \cdot q_2 = O\qty(q_1^2\qty(\log(q_1) + \log(M)\log(w))).
\end{equation}
\begin{table}[tb]
    \centering
    \caption{Number of quantum gates required by $\mathbf{A}_{y_i}$. \label{table:gate}}
    \begin{tabular}{lll}
        \hline
        Gate & Number of gates & \\
        \hline
        H & $q_{1}+q_{2}$ & $=O(q_{1} + \log M\log w)$\\
        R & $q_{2}$ & $=O(\log q_{1} + \log M \log w)$ \\
        1-CR & $q_{1} \cdot q_{2}$ & $=O\qty(q_{1}\qty(\log q_{1} + \log M \log w))$ \\
        2-CR & $q_{1}(q_{1}-1) \cdot q_{2} /2$ & $=O\qty(q_{1}^2\qty(\log q_{1} + \log M \log w))$ \\
        IQFT & 1 & $=O(1)$ \\
        \hline
    \end{tabular}
\end{table}
Based on the above analysis, the number of gates required to implement the state preparation operator $\mathbf{A}_{y_i}$ is summarized in Table~\ref{table:gate}.

\subsection{Specific Example}\label{subsec:example}

\begin{figure*}[tb]
\begin{equation}
\label{eq:QUBO_matrix}
\footnotesize
\mathbf{Q}=
\smqty[-176 & 64 & 64 & 64 & 64 & 33 & 64 & 33 & 33 & 33 & 33 & 32 & 64 & 33 & 33 & 33 & 33 & 32 & 64 & 64 & 33 & 33 \\
0 & -176 & 64 & 64 & 33 & 64 & 33 & 64 & 33 & 33 & 32 & 33 & 33 & 64 & 33 & 33 & 32 & 33 & 64 & 33 & 64 & 33 \\
0 & 0 & -176 & 33 & 64 & 64 & 33 & 33 & 64 & 32 & 33 & 33 & 33 & 33 & 64 & 32 & 33 & 33 & 33 & 64 & 64 & 33 \\
0 & 0 & 0 & -176 & 64 & 64 & 33 & 33 & 32 & 64 & 33 & 33 & 33 & 33 & 32 & 64 & 33 & 33 & 64 & 33 & 33 & 64 \\
0 & 0 & 0 & 0 & -176 & 64 & 33 & 32 & 33 & 33 & 64 & 33 & 33 & 32 & 33 & 33 & 64 & 33 & 33 & 64 & 33 & 64 \\
0 & 0 & 0 & 0 & 0 & -176 & 32 & 33 & 33 & 33 & 33 & 64 & 32 & 33 & 33 & 33 & 33 & 64 & 33 & 33 & 64 & 64 \\
0 & 0 & 0 & 0 & 0 & 0 & -176 & 64 & 64 & 64 & 64 & 33 & 64 & 33 & 33 & 33 & 33 & 32 & 64 & 64 & 33 & 33 \\
0 & 0 & 0 & 0 & 0 & 0 & 0 & -176 & 64 & 64 & 33 & 64 & 33 & 64 & 33 & 33 & 32 & 33 & 64 & 33 & 64 & 33 \\
0 & 0 & 0 & 0 & 0 & 0 & 0 & 0 & -176 & 33 & 64 & 64 & 33 & 33 & 64 & 32 & 33 & 33 & 33 & 64 & 64 & 33 \\
0 & 0 & 0 & 0 & 0 & 0 & 0 & 0 & 0 & -176 & 64 & 64 & 33 & 33 & 32 & 64 & 33 & 33 & 64 & 33 & 33 & 64 \\
0 & 0 & 0 & 0 & 0 & 0 & 0 & 0 & 0 & 0 & -176 & 64 & 33 & 32 & 33 & 33 & 64 & 33 & 33 & 64 & 33 & 64 \\
0 & 0 & 0 & 0 & 0 & 0 & 0 & 0 & 0 & 0 & 0 & -176 & 32 & 33 & 33 & 33 & 33 & 64 & 33 & 33 & 64 & 64 \\
0 & 0 & 0 & 0 & 0 & 0 & 0 & 0 & 0 & 0 & 0 & 0 & -176 & 64 & 64 & 64 & 64 & 33 & 64 & 64 & 33 & 33 \\
0 & 0 & 0 & 0 & 0 & 0 & 0 & 0 & 0 & 0 & 0 & 0 & 0 & -176 & 64 & 64 & 33 & 64 & 64 & 33 & 64 & 33 \\
0 & 0 & 0 & 0 & 0 & 0 & 0 & 0 & 0 & 0 & 0 & 0 & 0 & 0 & -176 & 33 & 64 & 64 & 33 & 64 & 64 & 33 \\
0 & 0 & 0 & 0 & 0 & 0 & 0 & 0 & 0 & 0 & 0 & 0 & 0 & 0 & 0 & -176 & 64 & 64 & 64 & 33 & 33 & 64 \\
0 & 0 & 0 & 0 & 0 & 0 & 0 & 0 & 0 & 0 & 0 & 0 & 0 & 0 & 0 & 0 & -176 & 64 & 33 & 64 & 33 & 64 \\
0 & 0 & 0 & 0 & 0 & 0 & 0 & 0 & 0 & 0 & 0 & 0 & 0 & 0 & 0 & 0 & 0 & -176 & 33 & 33 & 64 & 64 \\
0 & 0 & 0 & 0 & 0 & 0 & 0 & 0 & 0 & 0 & 0 & 0 & 0 & 0 & 0 & 0 & 0 & 0 & -176 & 64 & 64 & 64 \\
0 & 0 & 0 & 0 & 0 & 0 & 0 & 0 & 0 & 0 & 0 & 0 & 0 & 0 & 0 & 0 & 0 & 0 & 0 & -176 & 64 & 64 \\
0 & 0 & 0 & 0 & 0 & 0 & 0 & 0 & 0 & 0 & 0 & 0 & 0 & 0 & 0 & 0 & 0 & 0 & 0 & 0 & -176 & 64 \\
0 & 0 & 0 & 0 & 0 & 0 & 0 & 0 & 0 & 0 & 0 & 0 & 0 & 0 & 0 & 0 & 0 & 0 & 0 & 0 & 0 & -176 ]
\end{equation}
\normalsize
\hrulefill
\vspace*{4pt}
\end{figure*}

\begin{figure*}[tb]
    \centering
    \includegraphics[keepaspectratio,scale=0.7]{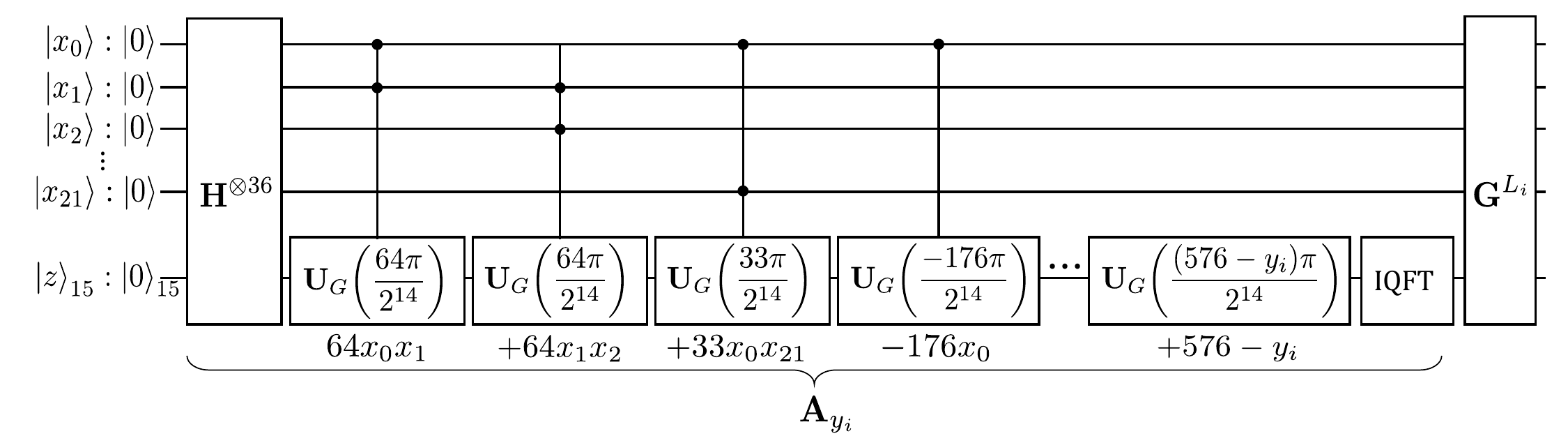}
    \caption{Quantum circuit corresponding to $E''(\mathbf{x}) = \mathbf{x}^{\mathrm{T}} \mathbf{Q} \mathbf{x} + 576 - y_i$.}
    \label{fig:circuit2}
\end{figure*}
Let us consider a specific example $(n, w, d, M) = (7, 3, 4, 7)$ of the proposed GAS.
In this case, the number of qubits is calculated as $q = q_1 + q_2'' = 22 + 15 = 37$ from \eqref{eq:bit_num1} and \eqref{eq:bit_num3}.
The objective function is $E''(\mathbf{x}) = \mathbf{x}^{\mathrm{T}} \mathbf{Q} \mathbf{x} + 576$, where the QUBO matrix $\mathbf{Q} \in \mathbb{Z}^{q_1 \times q_1}$ is given in \eqref{eq:QUBO_matrix}, and Fig.~\ref{fig:circuit2} shows the corresponding quantum circuit.
Based on the combinatorial matrix $\mathbf{P}(7, 3)$, the reduced matrix $\mathbf{P}'(7, 3) \in \mathbb{F}_2^{q_1 \times n}$ is generated by removing row vectors whose Hamming distance to $\mathbf{p}_0 = [1\ 1\ 1\ 0\ 0\ 0\ 0]$ is less than $d = 4$ as follows:
\begin{align}
&\mathbf{P}'(7, 3) \nonumber \\
=&
\mqty[
1~1~1~1~1~1~0~0~0~0~0~0~0~0~0~0~0~0~0~0~0~0 \\
0~0~0~0~0~0~1~1~1~1~1~1~0~0~0~0~0~0~0~0~0~0 \\
0~0~0~0~0~0~0~0~0~0~0~0~1~1~1~1~1~1~0~0~0~0 \\
1~1~1~0~0~0~1~1~1~0~0~0~1~1~1~0~0~0~1~1~1~0 \\
1~0~0~1~1~0~1~0~0~1~1~0~1~0~0~1~1~0~1~1~0~1 \\
0~1~0~1~0~1~0~1~0~1~0~1~0~1~0~1~0~1~1~0~1~1 \\
0~0~1~0~1~1~0~0~1~0~1~1~0~0~1~0~1~1~0~1~1~1
]^{\mathrm{T}}.
\end{align}
The proposed GAS outputs one of the optimal solutions
\begin{align}
\label{eq:x_opt}
\mathbf{x}_{\mathrm{opt}}  =
\mqty[
1 ~ 0 ~ 0 ~ 0 ~ 0 ~ 1 ~ 0 ~ 1 ~ 0 ~ 0 ~ 1 ~ 0 ~ 0 ~ 0 ~ 1 ~ 1 ~ 0 ~ 0 ~ 0 ~ 0 ~ 0 ~ 0
]^{\mathrm{T}},
\end{align}
which indicates that the optimal code consists of the first, sixth, eighth, 11th, 15th, and 16th rows of $\mathbf{P}'(7, 3)$ should be used, in addition to $\mathbf{p}_0 = [1\ 1\ 1\ 0\ 0\ 0\ 0]$.
Finally, the optimal code corresponding to \eqref{eq:x_opt} in a matrix format is
\begin{align}
\C_{\mathrm{opt}} = \mqty[1 & 1 & 1 & 0 & 0 & 0 & 0 \\
1 & 0 & 0 & 1 & 1 & 0 & 0 \\
1 & 0 & 0 & 0 & 0 & 1 & 1 \\
0 & 1 & 0 & 1 & 0 & 1 & 0 \\
0 & 1 & 0 & 0 & 1 & 0 & 1 \\
0 & 0 & 1 & 1 & 0 & 0 & 1 \\
0 & 0 & 1 & 0 & 1 & 1 & 0].
\end{align}

\section{Performance Comparisons}
In this section, we compare the proposed GAS with the conventional GAS in terms of query complexity and demonstrate that it converges to the optimal solutions faster.

\paragraph*{Evaluation Metrics}
In the following, the query complexity in the quantum domain corresponds to the total number of Grover operators, i.e., $\sum_i L_i$, and the query complexity in the classical domain corresponds to the number of measurements of the quantum states.
Both metrics are identical to those adopted in \cite{botsinis2014fixedcomplexity}.

\paragraph*{Considered Schemes}
The classic exhaustive search is considered as a reference scheme.
The original search space size is calculated as $\binom{\binom{n}{w}}{M}$.
Using Theorem~\ref{thm:uniformity} and $\mathbf{P}'(n, w)$, the search space size can be readily reduced to $\binom{q_1}{M-1}$, where the first codeword is fixed to $\mathbf{p}_0$ and the remaining $M-1$ codewords are selected from $q_1$ candidates.
The conventional GAS (Algorithm~\ref{alg:conv_GAS}) is considered as a performance baseline.
The proposed GAS (Algorithm~\ref{alg:proposed_GAS}) relies on the derived bounds, which contribute to increasing the speedup.
Both classical exhaustive search and conventional GAS used the proposed objective function $E'(\mathbf{x})$ of \eqref{eq:obj_func2}, while the proposed GAS used the objective function $E''(\mathbf{x})$ of \eqref{eq:obj_func3}.

\paragraph*{Simulation Parameters}
All the parameters are the same as those used in Section~\ref{subsec:example}.
The objective function values were averaged over $10^6$ trials.
The termination condition of GAS can be defined by, for example, the number of iterations, time, or threshold value.
Since the simulation was performed as a benchmark in this paper, we defined the termination condition that the threshold value $y_i$ reaches the minimum objective function value.
We assume the realization of a fault-tolerant quantum computer.

\begin{figure}[tb]
    \centering
    \subfigure[Classical domain.]{\includegraphics[keepaspectratio,scale=0.7]{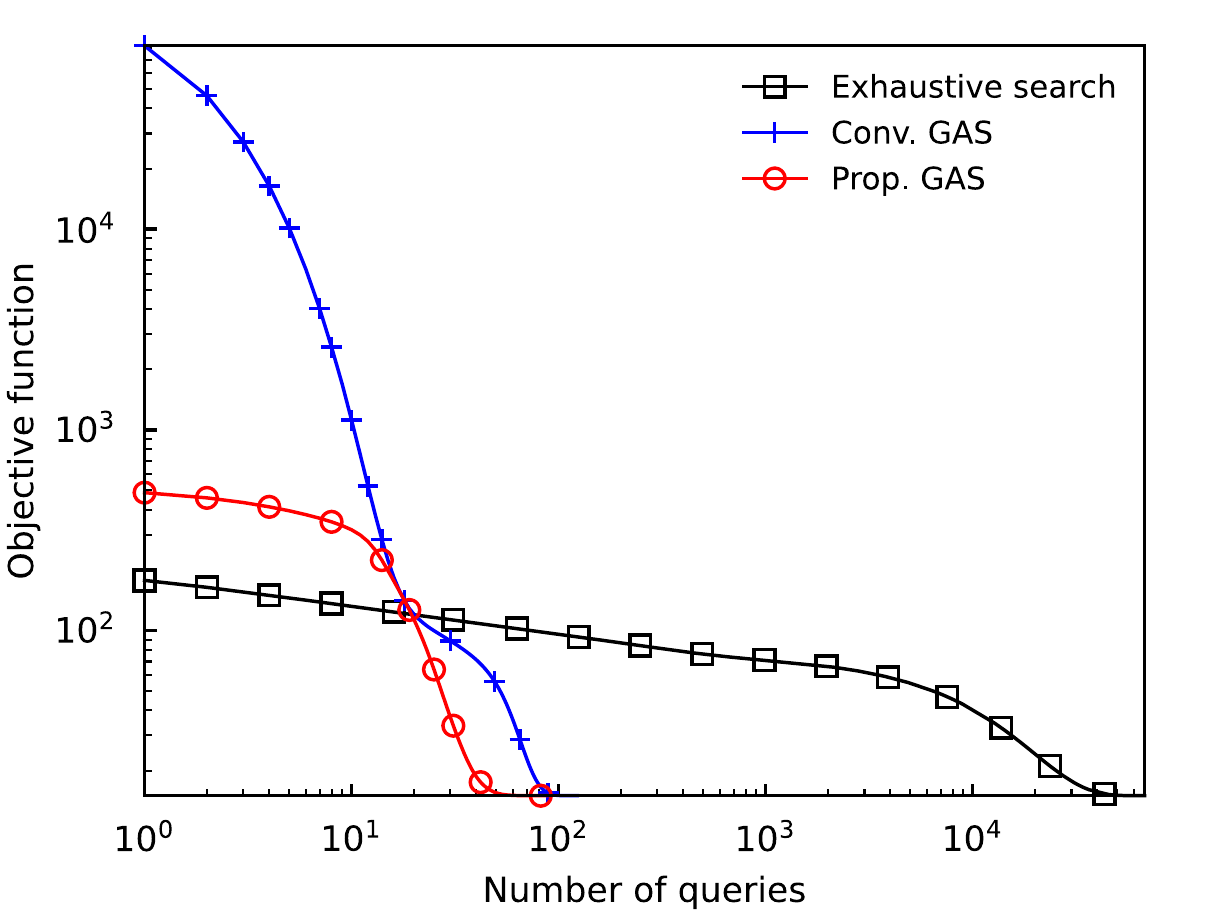}}
    \subfigure[Quantum domain.]{\includegraphics[keepaspectratio,scale=0.7]{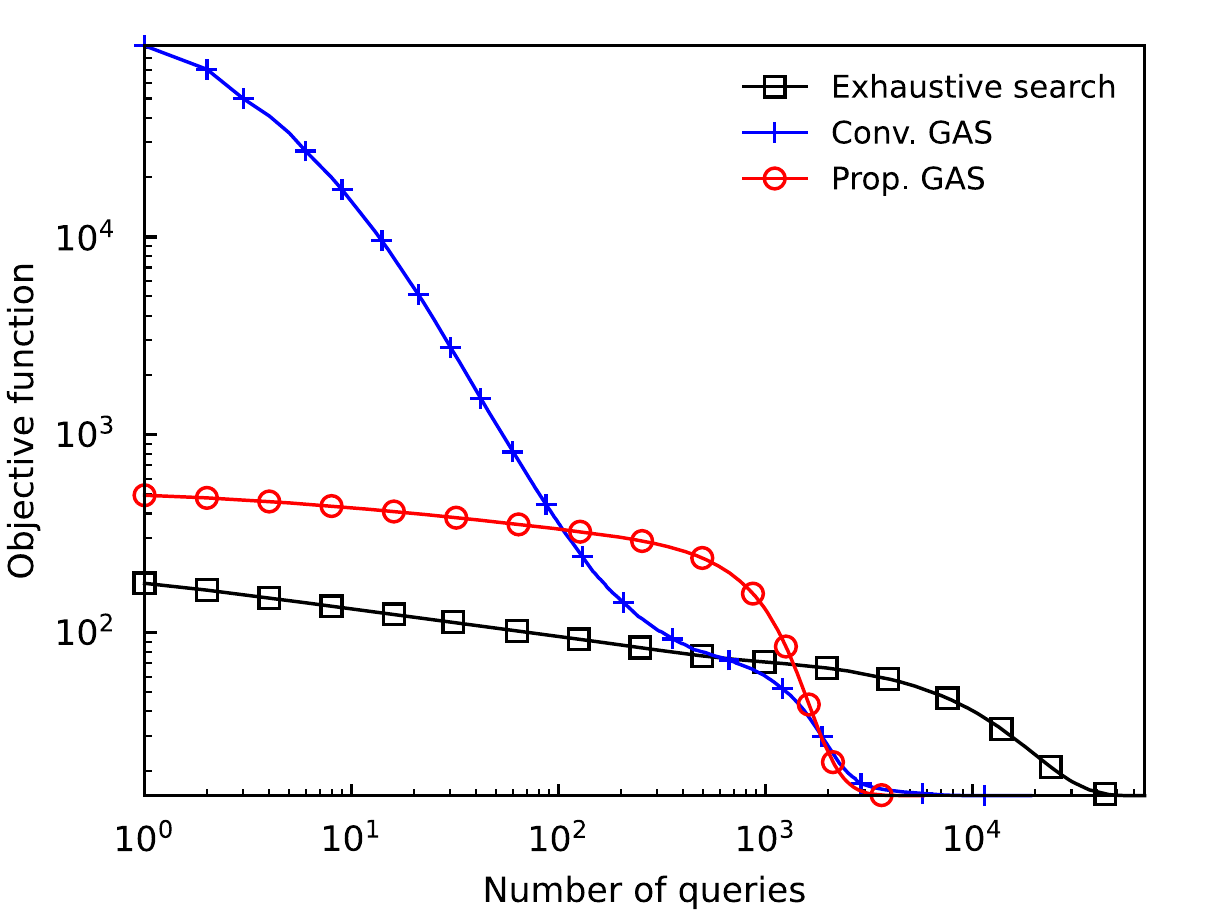}}
    \caption{Relationship between query complexity and the average of the objective function values, where we had $(n, w, d, M) = (7, 3, 4, 7)$.}
    \label{fig:obj_transition}
\end{figure}
First, Fig.~\ref{fig:obj_transition} shows the relationship between query complexity and the average of the objective function values, where query complexity was evaluated in both the classical and quantum domains.
As shown in Fig.~\ref{fig:obj_transition}, the objective function values for the conventional and proposed GAS differed significantly.
By suppressing the objective function value, the number of qubits was reduced from $q_2' = 22$ to $q_2'' = 15$.
Additionally, the proposed GAS exhibited the fastest convergence performance among the considered schemes in both the classical and quantum domains.
In regions where the number of queries was relatively small, the classical exhaustive search had smaller objective function values than those of the other GAS-based schemes.
This is because the proposed formulation generates a larger search space than the classical exhaustive search, i.e., $2^{q_1} > \binom{q_1}{M-1}$, which is a major drawback of our study.

The question arises ``why, despite this drawback, does the proposed GAS converge fast, especially in the classical domain?"
One major reason is the quadratic speedup of GAS with the aid of quantum superposition, entanglement, and amplitude amplification.
Another reason is that the proposed GAS starts with the strict initial threshold, $y_0 = \overline{E}'_{\mathrm{min}} + 1$ of \eqref{eq:y0Emin}, and the number of Grover operators $L_i$ is chosen from a limited range, $[0, \bar{k}_{\mathrm{opt}})$.
These proposals lead to fewer states of interest and higher probability of success, resulting in faster convergence.

In Fig.~\ref{fig:obj_transition}, we simply averaged $10^6$ curves for each scheme.
After converging to the minimum objective function value, which was 15, the value of each curve remained the same.
That is, in Fig.~\ref{fig:obj_transition}, the number of queries where each scheme reached the minimum was the maximum of observed $10^6$ query complexities.
For example, the classical exhaustive search reached the minimum objective function value at the query complexity of approximately $5.0 \cdot 10^4$, and this complexity was the maximum of all the trials.
The query complexity required to reach the minimum is an important metric and is a stochastic variable.
Its statistical property is investigated in Fig.~\ref{fig:CDF}.

\begin{figure}[tb]
    \centering
    \subfigure[Classical domain.]{\includegraphics[keepaspectratio,scale=0.7]{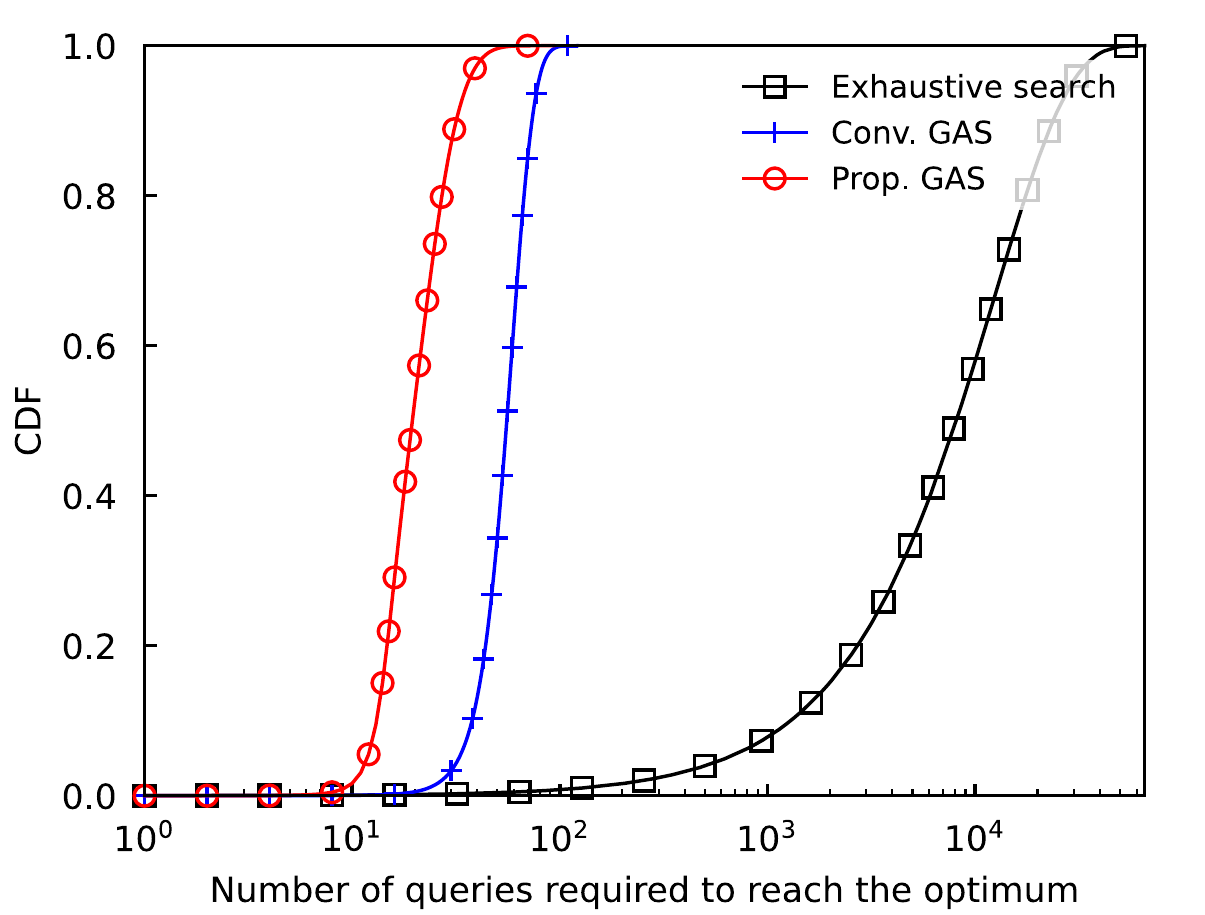}}
    \subfigure[Quantum domain.]{\includegraphics[keepaspectratio,scale=0.7]{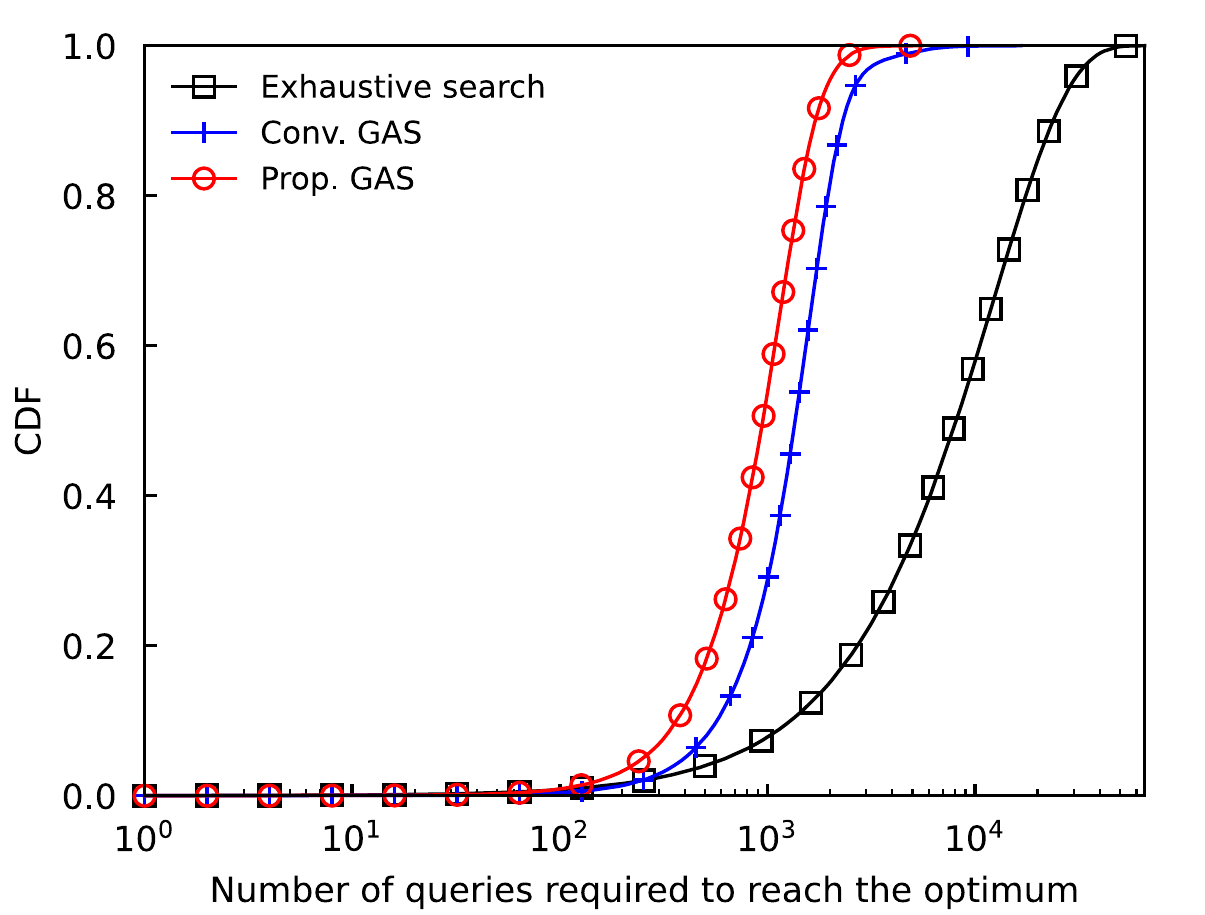}}
    \caption{CDF of query complexity required to reach the optimal solution, where we had $(n, w, d, M) = (7, 3, 4, 7)$.}
    \label{fig:CDF}
\end{figure}
In Fig.~\ref{fig:CDF}, the cumulative distribution function (CDF) of the query complexity required to reach the optimum was investigated.
In Fig.~\ref{fig:obj_transition}, the difference between the conventional GAS and the proposed GAS appears to be small, but when compared with CDF in Fig.~\ref{fig:CDF}, the difference is clear.
This result indicates that on average the proposed GAS reaches the optimal solution faster than the conventional one in both the classical and quantum domains.

\section{Conclusions}
Although much work has been done on binary constant weight codes, its parameter-dependent combinatorial explosion is still a problem. In this paper, we proposed a quantum search algorithm for binary constant weight codes as a solution.
The search problem was newly formulated as a QUBO problem, and novel bounds were derived for the minimum of the objective function value and the exact number of solutions.
We proposed a modified algorithm for GAS using the derived bounds and showed that it can further accelerate the quadratic speedup while reducing the number of qubits.
Similar approaches that utilize problem-specific bounds, such as those in this study, have the potential for wide use in applications of quantum search algorithms to other problems, thereby improving the feasibility of quantum speedup.

\section*{Acknowledgment}
IBM and Qiskit are trademarks of International Business Machines Corporation.

\appendices
\section{Proof of Theorem~\ref{THM:LOWER_BOUNDS_SOLUTION} \label{GILLIAM2021GROVER}}
We regard a code as a matrix $\mathbb{F}_2^{M \times n}$ where each row vector is a codeword containing $w$ 1s.
As mentioned in the proof of Theorem~\ref{thm:uniformity}, if we have a code that satisfies a given condition $(n, w, M, d)$, a new code generated by swapping any two columns also satisfies the same condition $(n, w, M, d)$.
Using this property, we prove that new codes can be generated by reordering the first $w$ columns of an existing code.
In the following, we state the lemmas necessary to prove Theorem~\ref{THM:LOWER_BOUNDS_SOLUTION}.
Note that $\mathbf{p}_0$ is the first row of the combinatorial matrix $\mathbf{P}(n, w)$ of \eqref{eq:Pnw}, and, hereinafter, it is assumed to be included in the code as a codeword without loss of generality.

\begin{lem}
If $A(n-1, d, w) < M$, then each column of a code that satisfies the condition $(n, w, M, d)$ contains at least one 1.
\end{lem}
\begin{proof}
Suppose that a code that satisfies the condition $(n, w, M, d)$ exists, and a column of that code contains only 0s.
Then, a code generated by removing that column satisfies the condition $(n-1, w, M, d)$.
The necessary and sufficient condition for such a code to exist is $A(n-1, d, w) \geq M$, which violates the assumption $A(n-1, d, w) < M$.
\end{proof}

\begin{lem}
\label{lem:two_1}
If $A(n-1, d, w) < M-1$, then each column of a code that satisfies the condition $(n, w, M, d)$ contains at least two 1s.
\end{lem}
\begin{proof}
According to Lemma~1, a column that contains only 0s in a code does not exist under the condition $(n-1, w, M, d)$.
Then, we prove that no column contains only one 1.
Suppose that there exists a code that contains only one 1 in a column.
A code generated by removing that column and that row, corresponding to the position of the contained 1, satisfies the condition $(n-1, w, M-1, d)$.
The necessary and sufficient condition for such a code is $A(n-1, d, w) \geq M-1$, which violates the assumption $A(n-1, d, w) < M-1$.
\end{proof}

\begin{lem}
\label{lem:lowerbound_of_1}
The first $w$ columns of any codeword except for $\mathbf{p}_0$ contain at most $w-\frac{d}{2}$ 1s.
\end{lem}
\begin{proof}
Except for $\mathbf{p}_0$, suppose that a codeword exists such that the first $w$ columns contain more 1s than $w-\frac{d}{2}$.
The inner product between that codeword and $\mathbf{p}_0$ is greater than $w-\frac{d}{2}$. From \eqref{eq:d_condition}, this violates the condition that the minimum distance of the code is at least $d$.
\end{proof}

\begin{cor}
\label{cor:upperbound_of_1}
The last $n-w$ columns of any codeword except for $\mathbf{p}_0$ contain at least $\frac{d}{2}$ 1s.
\end{cor}
\begin{proof}
Each codeword contains $w$ 1s. Thus, it is clear from Lemma~\ref{lem:lowerbound_of_1}.
\end{proof}
\begin{lem}
\label{lem:unexist}
If $w \leq d$, then there is no pair of codewords such that the last $n-w$ columns are identical.
\end{lem}
\begin{proof}
According to Corollary~\ref{cor:upperbound_of_1}, the last $n-w$ columns of any codeword except for $\mathbf{p}_0$ contain at least $\frac{d}{2}$ 1s.
If a pair of codewords exist, the last $n-w$ columns of which are identical, then $w - \frac{d}{2} > \frac{d}{2}$ holds, which violates the assumption $w \leq d$.
\end{proof}

\begin{lem}
\label{lem:special_case}
If $A(n-1, d, w) < M-1$ and $w - \frac{d}{2} = 1$, then $t \geq w!$ holds.
\end{lem}
\begin{proof}
According to Lemma~\ref{lem:two_1}, each column contains at least two 1s.
Also, according to Lemma~\ref{lem:lowerbound_of_1}, the number of 1s in the first $w$ columns is at most one except for $\mathbf{p}_0$.
Lemma~\ref{lem:two_1} and Lemma~\ref{lem:lowerbound_of_1} indicate that a code that satisfies the condition contains each column of $[\mathbf{I}_w ~ \mathbf{0}_{w, M-w}]^{\mathrm{T}} \in \mathbb{F}_2^{M \times w}$ in the first $w$ columns.
By reordering the first $w$ columns, $w!$ number of new distinguishable codes can be generated, and accordingly, the minimum number of solutions is calculated as $w!$.
\end{proof}

Based on the above lemmas, we prove Theorem~\ref{THM:LOWER_BOUNDS_SOLUTION}.
\setcounter{thm}{3}
\begin{thm}
If $w \leq d$ and $A(n-1,  d, w) < M-1$, then the exact number of solutions $t$ is lower bounded by
\begin{align}
    t &\geq \underline{t} = \left\{
    \begin{array}{ll}
        w! & (w-\frac{d}{2}=1) \\
        \displaystyle{\min_{2 \leq i \leq w - d/2}}\ \binom{w}{i} & (w-\frac{d}{2} \geq 2).
    \end{array}
    \right.
\end{align}
\end{thm}
\begin{proof}
According to Lemma~\ref{lem:special_case}, the $w - \frac{d}{2} = 1$ case is clear.
In the $w - \frac{d}{2} \geq 2$ case, each column of a code contains at least two 1s, according to Lemma~\ref{lem:two_1}.
Here, if the first $w$ columns of any codeword except for $\mathbf{p}_0$ contain at most one 1, non-overlapping $w!$ number of codes can be generated by reordering these columns as in Lemma~\ref{lem:special_case}.
Similarly, if the first $w$ columns contain $2 \leq i \leq w - \frac{d}{2}$ 1s, $\binom{w}{i}$ codes can be generated by reordering these columns.
Also, according to Lemma~\ref{lem:unexist}, since the last $n-w$ columns are all different, any generated code does not overlap with others.
That is, if $w-\frac{d}{2} \geq 2$, the following inequality holds:
\begin{align}
    t \geq \underline{t}  &= \min\qty{w!, \binom{w}{2}, \binom{w}{3}, ..., \binom{w}{w-\frac{d}{2}}} \nonumber\\
    &= \min\qty{\binom{w}{2}, \binom{w}{3}, ..., \binom{w}{w-\frac{d}{2}}} \nonumber\\
    &= \displaystyle{\min_{2 \leq i \leq w - d/2}}\ \binom{w}{i}.
\end{align}
\end{proof}

\section{Proof of Theorem~\ref{THM:UPPER_BOUND_K}}\label{app:thm_k}
\begin{proof}
Let the argument of \eqref{eq:k_opt} be
\begin{align}
    h(k) = \frac{4k^2\alpha}{2k\alpha - \sin(4k\theta)},
\end{align}
where we have $\alpha = \sin(2\theta) > 0$.
The domain of $h(k)$ is $k \in [1, \infty)$.
Then, to eliminate $\sin(4k\theta)$, we consider the following inequality
\begin{equation}
     \frac{4k^2\alpha}{2k\alpha + 1} \leq h(k) \leq  \frac{4k^2\alpha}{2k\alpha - 1}
\end{equation}
that holds for $k > 1/(2\alpha)$.
Next, for real numbers $k_1$ and $k_2$ that satisfies $1/(2\alpha) < k_1 < k_2$, we have an inequality
\begin{equation}
     \frac{4k_1^2\alpha}{2k_1\alpha - 1} \leq \frac{4k_2^2\alpha}{2k_2\alpha + 1}
\end{equation}
and we obtain
\begin{equation}
    k_2 \geq \underline{k}_2 = \frac{k_1^2\alpha + k_1\sqrt{k_1^2\alpha^2 + 2k_1\alpha - 1}}{2k_1\alpha - 1},
\end{equation}
which indicates that $k_2 \in [\underline{k}_2, \infty) \Rightarrow h(k_1) \leq h(k_2)$ holds for a given $k_1$.
Here, $k_1$ that minimizes $\underline{k}_2$ is
\begin{equation}
    \argmin_{k_1}\ \underline{k}_2 = \frac{1}{\alpha}.
\end{equation}
Then, the minimum value of $\underline{k}_2$ is expressed by
\begin{align}
    \label{eq:min_k2}
    \min_{k_1}(\underline{k}_2) &= \frac{1 + \sqrt{2}}{\alpha} = \frac{1 + \sqrt{2}}{2\sqrt{\frac{t}{2^{q_1}} - (\frac{t}{2^{q_1}})^2}} \nonumber\\
    &\sim \frac{1+\sqrt{2}}{2} \sqrt{\frac{2^{q_1}}{t}} < \frac{1+\sqrt{2}}{2} \sqrt{\frac{2^{q_1}}{\underline{t}}},
\end{align}
when $(t / 2^{q_1})^2$ is sufficiently small.
From \eqref{eq:min_k2}, $\overline{k}_{\mathrm{opt}}$ exists within $\qty[1, \left\lceil\frac{1+\sqrt{2}}{2} \sqrt{\frac{2^{q_1}}{\underline{t}}}\right\rceil]$.
\end{proof}

\footnotesize{
        \bibliographystyle{IEEEtran}
        \bibliography{main}

\begin{thebibliography}{10}
\providecommand{\url}[1]{#1}
\csname url@samestyle\endcsname
\providecommand{\newblock}{\relax}
\providecommand{\bibinfo}[2]{#2}
\providecommand{\BIBentrySTDinterwordspacing}{\spaceskip=0pt\relax}
\providecommand{\BIBentryALTinterwordstretchfactor}{4}
\providecommand{\BIBentryALTinterwordspacing}{\spaceskip=\fontdimen2\font plus
\BIBentryALTinterwordstretchfactor\fontdimen3\font minus
  \fontdimen4\font\relax}
\providecommand{\BIBforeignlanguage}[2]{{%
\expandafter\ifx\csname l@#1\endcsname\relax
\typeout{** WARNING: IEEEtran.bst: No hyphenation pattern has been}%
\typeout{** loaded for the language `#1'. Using the pattern for}%
\typeout{** the default language instead.}%
\else
\language=\csname l@#1\endcsname
\fi
#2}}
\providecommand{\BIBdecl}{\relax}
\BIBdecl

\bibitem{brouwer1990new}
A.~Brouwer, J.~Shearer, N.~Sloane, and W.~Smith, ``A new table of constant
  weight codes,'' \emph{IEEE Transactions on Information Theory}, vol.~36,
  no.~6, pp. 1334--1380, Nov. 1990.

\bibitem{tallini1998design}
L.~Tallini and B.~Bose, ``Design of balanced and constant weight codes for
  {{VLSI}} systems,'' \emph{IEEE Transactions on Computers}, vol.~47, no.~5,
  pp. 556--572, May 1998.

\bibitem{moon2005assignment}
J.~Moon, L.~Hughes, and D.~Smith, ``Assignment of frequency lists in frequency
  hopping networks,'' \emph{IEEE Transactions on Vehicular Technology},
  vol.~54, no.~3, pp. 1147--1159, May 2005.

\bibitem{chung1989optical}
F.~Chung, J.~Salehi, and V.~Wei, ``Optical orthogonal codes: Design, analysis
  and applications,'' \emph{IEEE Transactions on Information Theory}, vol.~35,
  no.~3, pp. 595--604, May 1989.

\bibitem{dodis2008fuzzy}
Y.~Dodis, R.~Ostrovsky, L.~Reyzin, and A.~Smith, ``Fuzzy extractors: {{How}} to
  generate strong keys from biometrics and other noisy data,'' \emph{SIAM
  Journal on Computing}, vol.~38, no.~1, pp. 97--139, Jan. 2008.

\bibitem{yassine2017index}
H.~Yassine, J.~P. Coon, and D.~E. Simmons, ``Index programming for flash
  memory,'' \emph{IEEE Transactions on Communications}, vol.~65, no.~5, pp.
  1886--1898, May 2017.

\bibitem{delacruz2021new}
B.~P. Dela~Cruz, J.~M. Lampos, H.~Palines, and V.~Sison, ``A new construction
  of anticode-optimal grassmannian codes,'' \emph{Journal of Algebra
  Combinatorics Discrete Structures and Applications}, vol.~8, no.~1, Jan.
  2021.

\bibitem{johnson1962new}
S.~Johnson, ``A new upper bound for error-correcting codes,'' \emph{IRE
  Transactions on Information Theory}, vol.~8, no.~3, pp. 203--207, Apr. 1962.

\bibitem{brouwer1980few}
A.~Brouwer, ``A few new constant weight codes,'' \emph{IEEE Transactions on
  Information Theory}, vol.~26, no.~3, pp. 366--366, May 1980.

\bibitem{nurmela1997new}
K.~Nurmela, M.~Kaikkonen, and P.~Ostergard, ``New constant weight codes from
  linear permutation groups,'' \emph{IEEE Transactions on Information Theory},
  vol.~43, no.~5, pp. 1623--1630, Sep. 1997.

\bibitem{smith2006new}
D.~H. Smith, L.~A. Hughes, and S.~Perkins, ``A new table of constant weight
  codes of length greater than 28,'' \emph{The Electronic Journal of
  Combinatorics}, vol.~13, no.~1, May 2006.

\bibitem{montemanni2009heuristic}
R.~Montemanni and D.~H. Smith, ``Heuristic algorithms for constructing binary
  constant weight codes,'' \emph{IEEE Transactions on Information Theory},
  vol.~55, no.~10, pp. 4651--4656, Oct. 2009.

\bibitem{etzion2014new}
T.~Etzion and A.~Vardy, ``A new construction for constant weight codes,'' in
  \emph{2014 {{International Symposium}} on {{Information Theory}} and Its
  {{Applications}}}, Oct. 2014, pp. 338--342.

\bibitem{shor1994algorithms}
P.~Shor, ``Algorithms for quantum computation: Discrete logarithms and
  factoring,'' in \emph{Proceedings 35th {{Annual Symposium}} on
  {{Foundations}} of {{Computer Science}}}, Nov. 1994, pp. 124--134.

\bibitem{grover1996fast}
L.~K. Grover, ``A fast quantum mechanical algorithm for database search,'' in
  \emph{Proceedings of the Twenty-Eighth Annual {{ACM}} Symposium on {{Theory}}
  of {{Computing}}}, {New York, NY, USA}, Jul. 1996, pp. 212--219.

\bibitem{BOYER1998TIGHT}
M.~Boyer, G.~Brassard, P.~H{\o}yer, and A.~Tapp, ``Tight bounds on quantum
  searching,'' \emph{Fortschritte der Physik}, vol.~46, no. 4-5, pp. 493--505,
  1998.

\bibitem{DURR1999QUANTUM}
C.~Durr and P.~Hoyer, ``A quantum algorithm for finding the minimum,'' Jan.
  1999.

\bibitem{BULGER2003IMPLEMENTING}
D.~Bulger, W.~P. Baritompa, and G.~R. Wood, ``Implementing pure adaptive search
  with {{Grover}}'s quantum algorithm,'' \emph{Journal of Optimization Theory
  and Applications}, vol. 116, no.~3, pp. 517--529, Mar. 2003.

\bibitem{BARITOMPA2005GROVER}
W.~P. Baritompa, D.~W. Bulger, and G.~R. Wood, ``Grover's quantum algorithm
  applied to global optimization,'' \emph{SIAM Journal on Optimization},
  vol.~15, no.~4, pp. 1170--1184, Jan. 2005.

\bibitem{GILLIAM2021GROVER}
A.~Gilliam, S.~Woerner, and C.~Gonciulea, ``Grover adaptive search for
  constrained polynomial binary optimization,'' \emph{Quantum}, vol.~5, p. 428,
  Apr. 2021.

\bibitem{norimoto2022quantum}
M.~Norimoto, R.~Mori, and N.~Ishikawa, ``Quantum speedup for higher-order
  unconstrained binary optimization and {{MIMO}} maximum likelihood
  detection,'' \emph{arXiv:2205.15478}, May 2022.

\bibitem{sano2022qubit}
Y.~Sano, M.~Norimoto, and N.~Ishikawa, ``Qubit reduction and quantum speedup
  for wireless channel assignment problem,'' \emph{arXiv:2208.05181}, Aug.
  2022.

\bibitem{frenger1999parallel}
P.~Frenger and N.~Svensson, ``Parallel combinatory {{OFDM}} signaling,''
  \emph{IEEE Transactions on Communications}, vol.~47, no.~4, pp. 558--567,
  Apr. 1999.

\bibitem{BRASSARD2002QUANTUM}
G.~Brassard, P.~Hoyer, M.~Mosca, and A.~Tapp, ``Quantum amplitude amplification
  and estimation,'' 2002, vol. 305, pp. 53--74.

\bibitem{shor1997polynomialtime}
P.~W. Shor, ``Polynomial-time algorithms for prime factorization and discrete
  logarithms on a quantum computer,'' \emph{SIAM Journal on Computing},
  vol.~26, no.~5, pp. 1484--1509, Oct. 1997.

\bibitem{qiskit}
M.~S. Anis, {Abby-Mitchell}, H.~Abraham \emph{et~al.}, ``Qiskit: {{An}}
  open-source framework for quantum computing,'' 2021.

\bibitem{Burden1989}
R.~L. Burden and J.~D. Faires, \emph{Numerical Analysis}, 4th~ed., ser. The
  Prindle, Weber and Schmidt Series in Mathematics.\hskip 1em plus 0.5em minus
  0.4em\relax {Boston}: {PWS-Kent Publishing Company}, 1989.

\bibitem{galantai2000theory}
A.~Gal{\'a}ntai, ``The theory of {{Newton}}'s method,'' \emph{Journal of
  Computational and Applied Mathematics}, vol. 124, no.~1, pp. 25--44, Dec.
  2000.

\bibitem{botsinis2014fixedcomplexity}
P.~Botsinis, S.~X. Ng, and L.~Hanzo, ``Fixed-complexity quantum-assisted
  multi-user detection for {{CDMA}} and {{SDMA}},'' \emph{IEEE Transactions on
  Communications}, vol.~62, no.~3, pp. 990--1000, Mar. 2014.

\end{thebibliography}
}


\end{document}